\documentclass[onecolumn,10pt]{IEEEtran}
\usepackage{graphicx}
\usepackage{amsmath,amssymb}
\usepackage{cases}
\usepackage{color}
\usepackage{amsfonts}
\usepackage{indentfirst}
\usepackage{cite}
\usepackage{caption}
\usepackage{algorithm}
\usepackage{algpseudocode}
\usepackage{booktabs}
\usepackage{subfigure}
\usepackage{multirow}
\usepackage{amsthm}
\usepackage{diagbox}
\makeatletter
\newtheoremstyle{mythm}{3pt}{3pt}{}{16pt}{\bfseries}{:}{.5em}{}
\theoremstyle{mythm}
\newtheorem{theorem}{Theorem}
\newtheorem{example}{Example}
\newtheorem{definition}{Definition}
\newtheorem{remark}{Remark}

\newtheorem{proposition}{Proposition}

\newtheorem{lemma}{Lemma}
\newtheorem{construction}{Construction}

\newcommand{\tabincell}[2]{\begin{tabular}{@{}#1@{}}#2\end{tabular}}

\begin{document}
\title{A framework of constructing placement delivery arrays for centralized coded caching
\author{Minquan Cheng, Jinyu Wang, Xi Zhong, Qiang Wang}
}

\date{}
\maketitle

\begin{abstract}
In caching system, it is desirable to design a coded caching scheme with the transmission load $R$ and subpacketization $F$ as small as possible, in order to improve efficiency of transmission in the peak traffic times and to decrease implementation complexity. Yan et al. reformulated the centralized coded caching scheme as designing a corresponding $F\times K$ array called placement delivery array (PDA), where $F$ is the subpacketization and $K$ is the number of users. Motivated by several constructions of PDAs, we introduce a  framework for constructing PDAs, where each row is indexed by a row vector of some matrix called row index matrix and each column's index is labelled by an element of a direct product set. Using this framework, a new scheme is obtained, which can be regarded as a generalization of some previously known schemes. When $K$ is equal to ${m\choose t}q^t$ for positive integers $m$, $t$  with $t<m$ and $q\geq 2$, we show that the row index matrix must be an orthogonal array if all the users have the same memory size. Furthermore, the row index matrix must be a covering array if the coded gain is ${m\choose t}$, which is the maximal coded gain under our framework.
Consequently the lower bounds on the transmission load and subpacketization of the schemes are derived under our framework. Finally, using orthogonal arrays as the row index matrix, we obtain two more explicit classes of schemes which have significantly advantages on the subpacketization while the transmission load is equal or close to that of the schemes constructed by Shangguan et al. (IEEE Trans. Inf. Theory, 64, 5755-5766, 2018) for the same number of users and memory size.

\end{abstract}

\section{Introduction}
\IEEEPARstart{T}{he} wireless network has been imposed a tremendous pressure on the data transmission during the peak traffic times due to the explosive increasing mobile services, especially the video streaming. Caching system has been recognized as an efficient solution to reduce such pressure. In a caching system, some contents are proactively placed into the users' memory during the off peak traffic times. Then the traffic amount could be reduced when the cached content is required by users during the peak traffic times. So traditionally
almost all of studies on caching system focused on exploiting the history or statistics of the user demands for an appropriate caching strategy. Maddah-Ali and Niesen \cite{MN} showed that coded caching can further reduce the traffic amount during the peak-traffic times by exploiting caches to create multicast opportunities.
\subsection{Centralized caching system}
In a centralized $(K,M,N)$ caching system (see Figure \ref{fig-origin-system}), a single server containing $N$ independent files with the same length connects to $K$ users over a shared link and each user has a cache memory of size $M$ files with $N \geq K$ and $N\geq M$. Denote the $N$ files by $\mathcal{W}=\{W_0,\ldots,W_{N-1}\}$ and $K$ users by $\mathcal{K}=\{0,1,\ldots,K-1\}$. For such a system, an $F$-division $(K,M,N)$ coded caching scheme consists of two phases as follows \cite{MN}:
\begin{itemize}
\item {\bf Placement phase:} During the off peak traffic times, each file is divided into $F$ equal packets\footnote{\label{foot-meory-sharing}Memory sharing technique may lead to non equally divided packets \cite{MN}, in this paper, we will not discuss this case.}, i.e., $W_{i}=\{W_{i,j}:j=0,1,\ldots,F-1\}$. Then each user caches some packets (or linear combinations of packets) from the server. If packets are cached directly, it is called uncoded placement; if linear combinations of packets are cached, we call it coded placement. Denote $\mathcal{Z}_k$ the contents cached by user $k$. In this phase we assume that the server does not know the users' requests in the following phase. For simplicity, most known results are proposed under the assumption of {\em identical uncoded caching policy}, i.e., if some user caches the $i$-th packet of some file, then it caches the $i$-th packet of all files.

\item {\bf Delivery phase:} During the peak traffic times, each user randomly requests one file from the files set $\mathcal{W}$ independently. The request vector is denoted by $\mathbf{d}=(d_0,\cdots,d_{K-1})$, i.e., user $k$ requests the $d_k$-th file $W_{d_k}$, where $d_k\in[0,N)$ and $k\in\mathcal{K}$. The server broadcasts a coded signal (XOR of some required packets) of size $S_{{\bf d}}$ packets to users such that each user is able to recover its requested file with the help of its cache contents.
\end{itemize}
\begin{figure}[h]
\centering
\includegraphics[width=3in]{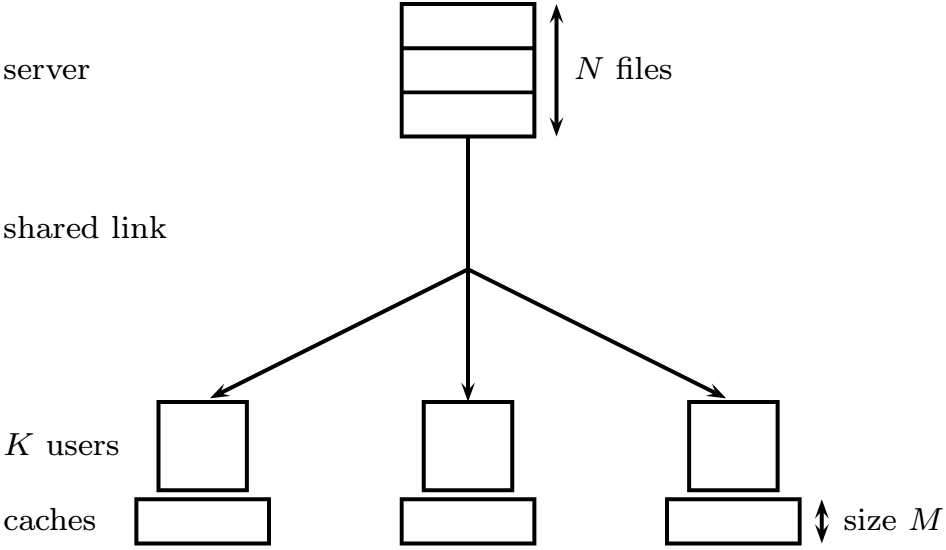}
\vskip 0.2cm
\caption{Centralized $(K,M,N)$ caching system}\label{fig-origin-system}
\end{figure}
In this paper, we focus on the worst-case demand pattern, i.e., all the users require different files. In this case,  the transmission load of a coded caching scheme is defined as the maximal normalized transmission amount among all the requests in the delivery phase, i.e.
$$R=\max_{{\bf d}\in\{0,\ldots,N-1\}^K }\left\{\frac{S_{{\bf d}}}{F}\right\}.$$ Since the implementation complexity of a coded caching scheme increases along with its subpacketization $F$, it is desirable to design a scheme with the transmission load and the subpacketization as small as possible.
\subsection{Prior work}
In this paper, we are interested in coded caching schemes with uncoded placement for a centralized caching system. There are many related studies such as \cite{MN,AG,GR,WTP,YMA,JCLC,STC,TC,SJTLD,CJTY,YCTC,TR,SZG,STD,YTCC,K,ASK,CJWY}. These results focus on either the transmission load or the subpacketization.

If each file in the library is large enough and $\frac{KM}{N}$ is an integer, Maddah-Ali and Niesen \cite{MN} proposed the first coded caching scheme with subpacketization $F={K\choose KM/N}$ and obtained a lower bound on the transmission load. This scheme is referred to as the MN scheme, which is listed in the second row of Table \ref{tab-known}. In \cite{GR}, an improved lower bound of the transmission load was derived  by Ghasemi and  Ramamoorthy using a combinatorial problem of optimally labeling the leaves of a directed tree. Wan et al. \cite{WTP} showed that the MN scheme has minimum transmission load by the graph theory when $K<N$; Yu et al. \cite{YMA} obtained the transmission load for various demand patterns by modifying the delivery phase of the MN scheme. There are some studies focusing on the transmission load under other conditions. For example, Jin. et al. \cite{JCLC} derived the average transmission load for conditions of non-uniform file popularity and various demand patterns respectively, using an optimization framework under a specific caching rule. When $K<N$ and all files have the same popularity, they also showed that the minimum transmission load is exactly that of the MN scheme.

It is well known that there is a tradeoff between $R$ and $F$. To the best of our knowledge, most known results focusing on the subpacketization are proposed under identical uncoded caching policy. When $K$ is large, the first tradeoff between $F$ and $R$ was derived by Shanmugam et al. in \cite{SJTLD} using probabilistic arguments. For an abitrary positive integer $K$, it is  still hard to derive the tradeoff between $F$ and $R$. Hence most  previously known constructions focus on reducing the subpacketization. In \cite{CJTY} Cheng et al. showed that, given the minimum transmission load, the minimum subpacketization is $F={K\choose KM/N}$, i.e., the subpacketization of the MN scheme. It is easy to see that the subpacketization $F$ of the MN scheme increases very quickly when the number of users $K$ increases. The first scheme with lower supacketization compared with the MN scheme was proposed by Shanmugam et al. in \cite{SJTLD} by a grouping method. Yan et al. in \cite{YCTC} generated two classes of schemes by means of a combinatorial structure called placement delivery array (PDA), which is an $F\times K$ array and could be used to characterize the placement and delivery phases for all possible requests. There are some other coded caching schemes with lower subpacketization levels obtained by increasing the transmission load $R$, for instance, schemes by means of linear block codes \cite{TR}, the special $(6,3)$-free hypergraphs \cite{SZG}, the $(r,t)$ Ruzsa-Szem\'{e}redi graphs \cite{STD}, the strong edge coloring of bipartite graphs \cite{YTCC},  the projective space over $\mathbb{F}_p$ \cite{K}, combinatorial design \cite{ASK} and so on. In \cite{CJYT} the authors generalized the constructions in \cite{TR} and proposed two classes of schemes with flexible memory size based on the constructions in \cite{YCTC,SZG} using the generalized grouping method in \cite{CJWY}. By means of the MN scheme, the authors in \cite{CJTY} obtained two variant MN schemes where $F$ is minimal for the fixed $R$ and memory size $M$. Some of these schemes are listed in Table \ref{tab-known}, which was summarized by Krishnan in \cite{K}.
\begin{table}[http!]
\center
\caption{Summary of some known results\label{tab-known}}
\small{
\begin{tabular}{|c|c|c|c|c|c|}
\hline
Schemes and parameters &User number $K$  &Caching ratio $\frac{M}{N}$   &Load $R$   &Subpacketization $F$    \\ \hline
\tabincell{l}{MN scheme in \cite{MN},  any $k$, $t\in  \mathbb{Z}^+$ with $t<k$} &$k$&$\frac{t}{k}$&$\frac{k-t}{1+t}$&${k\choose t}$ \\ \hline

\tabincell{l}{Scheme in \cite{YCTC}, any $m$ and $q\in\mathbb{Z}^+$}&\multirow{2}*{$(m+1)q$}&$\frac{1}{q}$&$q-1$&${q}^m$\\[0.3cm] \hline

\tabincell{l}{Scheme in \cite{SZG}, any $m$, $t$, $q\in\mathbb{Z}^+$ with $t<m$}&\multirow{2}*{${m\choose t}{q}^{t}$} &$1-(\frac{q-1}{q})^{t}$&$(q-1)^{t}$&${q}^m$\\[0.3cm]  \hline

\tabincell{l}{Scheme in \cite{CJTY}, any $k$, $t\in  \mathbb{Z}^+$ with $t<k$}
&${k\choose t+1}$ &$1-\frac{t+1}{{k\choose t}}$&$\frac{k}{{k\choose t}}$&${k\choose t}$\\   \cline{2-5}
&${k\choose t}$ &$\frac{t}{k}$&$\frac{{k\choose t+1}}{k}$& $k$\\   \hline

\tabincell{l}{Scheme in \cite{YTCC}, any $m$, $a$, $b$, $\lambda\in\mathbb{Z}^+$\\ with $a,b<m$ and $\lambda<\min \{a,b\}$ }&${m\choose a}$ &$1-\frac{{a\choose \lambda}{m-a\choose b-\lambda}}{{m\choose b}}$ &
$\frac{{m\choose a+b-2\lambda}{a+b-2\lambda\choose a-\lambda}}{{m\choose b}}$
&${m\choose b}$\\   \hline
\end{tabular} }
\end{table}

\subsection{Contribution and organization  of this paper}

In \cite{SDLT}, Shanmugam et al. discovered that most of coded caching schemes can be recasted into a PDA when $K\leq N$. This implies that PDA is a useful tool for designing coded caching schemes. It is worth noting that PDA can also be used to study distributed computing \cite{YYW,YWYT}, privacy coded caching \cite{CLW,AST}, linear function retrieval \cite{YT} and so on. In this paper, we will consider the schemes by means of PDAs when $K\leq N$.

Motivated by several constructions of PDA, in this paper we introduce a new framework for constructing PDA, where each row index is denoted by a row vector of some matrix called row index matrix and each column index is denoted by an element of a direct product set called column index set. Then the problem of designing a PDA (i.e., coded caching scheme) can be transformed into a problem of choosing a row index matrix and a column index set appropriately. It is worth noting that some previously known schemes (e.g., the MN scheme \cite{MN}, the first PDA in \cite{YCTC}, PDAs in \cite{SZG}) can be obtained under our framework. For $K={m\choose t}q^t$ for  positive integers $m$, $t$ with $t<m$ and $q\geq 2$, we obtain the following results:
\begin{itemize}
\item If all the users have the same memory size, the row index matrix must be an orthogonal array with degree $m$, $q$ levels and strength $t$.
\item If the coded gain is ${m\choose t}$, which is the maximal coded gain under our framework, the row index matrix must be a covering array with degree $m$, $q$ levels and strength $m-t$. Consequently a lower bound on the subpacketization of the schemes with the maximal coded gain ${m\choose t}$ is obtained, i.e., $F\geq q^{m-t}$.
\end{itemize}
Here orthogonal arrays and covering arrays are classic combinatorial structures in combinatorial design theory \cite{Stinson} and they will be introduced in detail in Section \ref{sect-construction via OA}.
The above results imply the following two statements:
\begin{itemize}
\item When $K={m\choose t}q^t$, if we want to design a coded caching scheme with the maximal coded gain $g={m\choose t}$, we only need to choose an appropriate row index matrix which is not only an orthogonal array with strength $t$ but also a covering array with strength $m-t$.
\item The scheme in \cite{YCTC} listed in the third row of Table \ref{tab-known} achieves our lower bound on the subpacketization.
\end{itemize}

In addition, we present three new constructions of PDA under our framework, which are listed in Table \ref{tab-main}. It is worth noting that compared with the scheme from \cite{SZG} (see Table \ref{tab-known}), the second construction reduces the subpacketization by a factor of $q$ without sacraficing the transmission load and the third construction reduces the subpacketization by a factor of $q^t$ at a cost of transmission load increasing from $(q-1)^t$ to $q^t -1$ for the same number of users and the same memory size.

\begin{table}[http!]
\center
\caption{Some main schemes in this paper\label{tab-main}}
\small{
\begin{tabular}{|c|c|c|c|c|c|}
\hline
Schemes and parameters &User number $K$  &Caching ratio $\frac{M}{N}$   &Load $R$   &Packet number $F$    \\ \hline
\tabincell{l}{Theorem \ref{th-subsets-PDA},  any $m$, $s$, $t\in  \mathbb{Z}^+$, $\omega\in \mathbb{Z}$\\ with $0\leq \omega\leq t\leq s$, $s+t-2\omega\leq m$
} &${t\choose \omega}{m\choose t}$&$1-\frac{{m-t\choose s-w}}{{m\choose s}}$&${m\choose s+t-2\omega}/{m\choose s}$&${m\choose s}$ \\ \hline

\tabincell{l}{Theorem \ref{th-s=m-1}, any $m$, $s$, $t$, $q\in\mathbb{Z}^+$,\\ $0<t<m$} & ${m\choose t}q^t$&$1-(\frac{q-1}{q})^{t}$ & $(q-1)^t$ & $q^{m-1}$\\  \hline

\tabincell{l}{Theorem \ref{th-reult-m-2}, any $m,t\in\mathbb{Z}^+$ with\\ $2t\leq m$, and some prime power $q$ }&${m\choose t}q^t$ &$1-(\frac{q-1}{q})^t$ & $q^t-1$ & $q^{m-t}$\\[0.3cm]   \hline

\end{tabular} }
\end{table}

The rest of this paper is organized as follows.

In Section \ref{sec_prob} we introduce basic notations, the definition of PDA and the
relationship between a PDA and a coded caching scheme. Then in Section \ref{se-characterization} a framework for constructing PDA is proposed
and a new scheme is obtained, which can be regarded as a generalization of several previously known schemes. In Section  \ref{sect-construction via OA}, the requirements for the row index matrix is discussed when each user has the same memory size or the coded gain is the
largest. Another two new schemes based on orthogonal arrays are obtained in Section \ref{sec-OAs}. Finally the conclusion is drawn in
Section \ref{conclusion}.

\section{Preliminaries}\label{sec_prob}
In this paper, we will use the following notations unless otherwise stated.
\begin{itemize}
\item We use bold capital letters,  bold lower case letters and curlicue letters to denote arrays, vectors and sets respectively.
\item For any positive integers $m$ and $t$ with $t< m$, let $[0,m)=\{0,1,\ldots,m-1\}$ and ${[0,m)\choose t}=\{\mathcal{T}\ |\   \mathcal{T}\subseteq [0,m), |\mathcal{T}|=t\}$, i.e., ${[0,m)\choose t}$ is the collection of all $t$-sized subsets of $[0,m)$.
\item Given an $F\times m$ matrix $\mathbf{F}$ and a subset $\mathcal{T}\subseteq [0,m)$, $\mathbf{F}|_{\mathcal{T}}$ is a submatrix obtained by taking only the columns with subscript $j \in \mathcal{T}$ (i.e., deleting all the columns with subscript $j\in [0,m)\setminus \mathcal{T}$). Similarly,  given a vector ${\bf a}$ with length $m$, ${\bf a}|_{\mathcal{T}}$ is a vector with length $|\mathcal{T}|$ obtained by taking only the coordinates with subscript $j \in \mathcal{T}$.
\item For any two vectors ${\bf x}$ and ${\bf y}$ with the same length, $d({\bf x},{\bf y})$ is the number of coordinates in which ${\bf x}$ and ${\bf y}$ differ. Let $wt({\bf x})$ be the weight of ${\bf x}$, i.e., the number of nonzero coordinates of ${\bf x}$.
\end{itemize}

\subsection{Placement delivery array}
Yan et al.  \cite{YCTC} proposed a combinatorial structure, called placement delivery array, which can characterize the placement phase and delivery phase simultaneously.
\begin{definition}\rm(\cite{YCTC})
\label{def-PDA}
For positive integers $K$,  $F$, and $S$, an $F\times K$ array $\mathbf{P}=(p_{j,k})$, $j\in [0,F), k\in[0,K)$, composed of a specific symbol $``*"$ called star and $S$ symbols $0,1,\cdots, S-1$, is called a $(K,F,S)$ placement delivery array (PDA) if it satisfies the following condition C$1$:
\begin{enumerate}
  \item [C$1$.] For any two distinct entries $p_{j_1,k_1}$ and $p_{j_2,k_2}$,    $p_{j_1,k_1}=p_{j_2,k_2}=s$ is an integer  only if
  \begin{enumerate}
     \item [a.] $j_1\ne j_2$, $k_1\ne k_2$, i.e., they lie in distinct rows and distinct columns; and
     \item [b.] $p_{j_1,k_2}=p_{j_2,k_1}=*$, i.e., the corresponding $2\times 2$  subarray formed by rows $j_1,j_2$ and columns $k_1,k_2$ must be of the following form
  \begin{eqnarray*}
    \left(\begin{array}{cc}
      s & *\\
      * & s
    \end{array}\right)~\textrm{or}~
    \left(\begin{array}{cc}
      * & s\\
      s & *
    \end{array}\right).
  \end{eqnarray*}
   \end{enumerate}
  \end{enumerate}
\end{definition}
In addition, for any positive integer $Z\leq F$,  the array $\mathbf{P}$ is denoted by $(K,F,Z,S)$ PDA if the following condition is further satisfied:
\begin{enumerate}
\item [C$2$.] each column has exactly $Z$ stars.
   \end{enumerate}
For instance, it is easy to verify that the following array is a $(6,4,2,4)$ PDA.
\begin{eqnarray}
\label{eq-E-pda-1}
\mathbf{P}=\left(\begin{array}{cccccc}
*&*&*&0&1&2\\
*&0&1&*&*&3\\
0&*&2&*&3&*\\
1&2&*&3&*&*
\end{array}\right).
\end{eqnarray}
In \cite{YCTC}, Yan et al. showed that a PDA can generate a coded caching scheme as follows.
\begin{theorem}(\cite{YCTC})
\label{th-Fundamental}If there exits a $(K,F,Z,S)$ PDA, then we can obtain an $F$-division $(K,M,N)$ coded caching scheme with $\frac{M}{N}=\frac{Z}{F}$ and transmission load $R=\frac{S}{F}$.
\end{theorem}

In particular, Algorithm \ref{alg:PDA} was used  to generate a coded caching scheme from a PDA. In fact, a PDA $\mathbf{P}=(p_{j,k})$ is an $F\times K$ array composed of a specific symbol $``*"$ and $S$ integers, where columns represent the user index and rows represent the packet index. If $p_{j,k}=*$, then user $k$ has cached the $j$-th packet of all files, which satisfies identical uncoded caching policy. If $p_{j,k}=s$ is an integer, it means that the $j$-th packet of all files is not stored by user $k$. Then the server broadcasts a coded packet (i.e. the XOR of all the requested packets indicated by $s$) to the users at time slot $s$. The property C$1$ of PDA  guarantees that each user can get its requested packet, since it has cached all the other packets in the coded packet except its requested one. The occurrence number of the integer $s$ in $\mathbf{P}$, denoted by $g_s$, is the coded gain at time slot $s$, since the coded packet broadcasted at time slot $s$ is useful for $g_s$ users.
\begin{algorithm}[htb]
\caption{Coded caching scheme based on PDA in \cite{YCTC}}\label{alg:PDA}
\begin{algorithmic}[1]
\Procedure {Placement}{$\mathbf{P}$, $\mathcal{W}$}
\State Split each file $W_n\in\mathcal{W}$ into $F$ packets, i.e., $W_{n}=\{W_{n,j}\ |\ j\in[0,F)\}$.
\For{$k\in \mathcal{K}$}
\State $\mathcal{Z}_k\leftarrow\{W_{n,j}\ |\ p_{j,k}=*, \forall~n\in [0,N)\}$
\EndFor
\EndProcedure
\Procedure{Delivery}{$\mathbf{P}, \mathcal{W},{\bf d}$}
\For{$s=0,1,\cdots,S-1$}
\State  Server sends $\bigoplus_{p_{j,k}=s,j\in[0,F),k\in[0,K)}W_{d_{k},j}$.
\EndFor
\EndProcedure
\end{algorithmic}
\end{algorithm}

\begin{example}\rm
\label{E-pda}
Using the PDA in \eqref{eq-E-pda-1} and Algorithm \ref{alg:PDA},  we can obtain a $4$-division $(6,3,6)$ coded caching scheme as follows.
\begin{itemize}
   \item \textbf{Placement Phase}: From Line 2 of  Algorithm \ref{alg:PDA}, we have $W_n=\{W_{n,0},W_{n,1},W_{n,2},W_{n,3}\}$, $n\in [0,6)$. Then by Lines 3-5, the contents in each user are
\begin{eqnarray*}
       \mathcal{Z}_0=\left\{W_{n,0},W_{n,1}\ |\ n\in[0,6)\right\},
       \mathcal{Z}_1=\left\{W_{n,0},W_{n,2}\ |\ n\in[0,6)\right\},\\
       \mathcal{Z}_2=\left\{W_{n,0},W_{n,3}\ |\ n\in[0,6)\right\},
       \mathcal{Z}_3=\left\{W_{n,1},W_{n,2}\ |\ n\in[0,6)\right\},\\
       \mathcal{Z}_4=\left\{W_{n,1},W_{n,3}\ |\ n\in[0,6)\right\},
       \mathcal{Z}_5=\left\{W_{n,2},W_{n,3}\ |\ n\in[0,6)\right\}.
\end{eqnarray*}
   \item \textbf{Delivery Phase}: Assume that the request vector is $\mathbf{d}=(0,1,2,3,4,5)$. By Lines 8-10, the signals sent by the server are listed in Table \ref{table1}.
   \begin{table}[!htp]
\centering
  \caption{Delivery steps in Example \ref{E-pda} }\label{table1}
  \small{
  \begin{tabular}{|c|c|}
\hline
   Time Slot& Transmitted Signnal\\ \hline
   $0$&$W_{0,2}\oplus W_{1,1}\oplus W_{3,0}$\\ \hline
   $1$&$W_{0,3}\oplus W_{2,1}\oplus W_{4,0}$\\ \hline
  $2$& $W_{1,3}\oplus W_{2,2}\oplus W_{5,0}$\\ \hline
  $3$& $W_{3,3}\oplus W_{4,2}\oplus W_{5,1}$\\ \hline
  \end{tabular}}
\end{table}
\end{itemize}
\end{example}
\section{A framework for constructing placement delivery array}
\label{se-characterization}
Motivated by several constructions of PDAs, we propose a framework for constructing PDAs as follows.
\begin{construction}
\label{construction}
For any positive integers $q$, $m$ and $t$ with $0<t<m$, let
$$
\mathcal{F}\subseteq [0,q)^m,  \ \ \mathcal{K}\subseteq {[0,m)\choose t}\times [0,q)^t,
$$
then an $|\mathcal{F}|\times |\mathcal{K}|$ array $\mathbf{P}=(p_{{\bf f},{\bf k}})$ with ${\bf f}=(f_0,f_1,\ldots, f_{m-1})\in \mathcal{F}$ and  ${\bf k}=({\mathcal{T}},{\bf b})=(\{\delta_0,\delta_1,\ldots, \delta_{t-1}\},(b_0,b_1,\ldots,b_{t-1}))\in \mathcal{K}$
with $0\leq\delta_0<\ldots<\delta_{t-1}<m$, can be defined in the following way
\begin{eqnarray}
\label{eq-constr-PDA}
p_{{\bf f},{\bf k}}=\left\{
\begin{array}{ll}
({\bf e},n_{\bf e}) & \textrm{if}~d({\bf f}|_{\mathcal{T}}$, ${\bf b})=t, \\
* & \textrm{otherwise},
\end{array}
\right.
\end{eqnarray}
where ${\bf e}=(e_0,e_{1},\ldots,e_{m-1})\in[0,q)^m$ such that  \begin{eqnarray}
\label{eq-putting integer}
e_i=\left\{
\begin{array}{ll}
b_h & \textrm{if}\ i=\delta_h, h\in [0,t),\\[0.2cm]
f_i & \textrm{otherwise.} \end{array}
\right.
\end{eqnarray} and $n_{\bf e}$ is the  occurrence  order of vector ${\bf e}$ that occurs in column ${\bf k}$ and starts from $0$. For convenience, the matrix ${\bf F}$, whose collection of row vectors is $\mathcal{F}$, is called the row index matrix. The set $\mathcal{K}$ is called the column index set.
\end{construction}

We will show that the array constructed by Construction \ref{construction} is a PDA later. It is worth noting that in \cite{SZG} the authors constructed a PDA, where $\mathcal{F}=[0,q)^m$, $\mathcal{K}={[0,m)\choose t}\times [0,q)^t$ and the non-star entries are represented by $(m+t)$-dimensional vectors. Clearly the difference between the construction in \cite{SZG} and our construction lies in the representations of the row index matrix, column index set and definition of non-star entries. In our construction, the row index matrix  and column index set are more flexible and the rule of non-star entries is simpler compared with the construction in \cite{SZG}. Moreover, schemes with better performance can be obtained by Construction \ref{construction}.

\begin{example}
\label{ex-OA423-Array}
Assume that $q=2$, $m=3$ and $t=2$. Let $\mathcal{F}=\{(0,0,0), (1,0,1),(0,1,1),(1,1,0)\}$, i.e.,
\begin{eqnarray}
\label{eq-OA-1}\mathbf{F}=\left(
\begin{array}{ccc}
0&0&0\\
1&0&1\\
0&1&1\\
1&1&0
\end{array}\right),
\end{eqnarray}
$\mathcal{K}={[0,3)\choose 2}\times [0,2)^2$. Then we have $F=|\mathcal{F}|=4$ and $K=|\mathcal{K}|={3\choose 2}\times 4=12$. From Construction \ref{construction}, the following $4\times 12$ array can be obtained.
\begin{eqnarray}
\label{eq-PDA-322}\begin{footnotesize}
\begin{array}{|c|cccc|cccc|cccc|}\hline
\multicolumn{1}{|c|}{ \multirow{2}*{ \diagbox{${\bf f}$} {$(\mathcal{T},{\bf b} )$}} }&
 \multicolumn{4}{|c|}{ \multirow{1}*{$\mathcal{T}=\{0 ,1\}$}} & \multicolumn{4}{|c|}{ \multirow{1}*{$\mathcal{T}=\{0,2\}$}}& \multicolumn{4}{|c|}{ \multirow{1}*{$\mathcal{T}=\{1,2\}$}}\\ \cline{2-13}
       &00  & 10 & 01 & 11 & 00 & 10&01   &11 & 00  & 10 & 01 &11\\ \hline
000&  *    &*      &*      &1100&*      &*      &*      &1010&*      &*      &*      &0110\\
101&  *    &*      &0110&*      &0000&*      &*      &*      &*      &1100&*      &*\\
011&*      &1010&*      &*      &*      &1100&*      &*      &0000&*      &*      &*\\
110&0000&*      &*      &*      &*      &*      &0110&*      &*      &*      &1010&*\\ \hline
\end{array}\end{footnotesize}
\end{eqnarray}
Let us consider the column represented by $(\{0,1\},(0,0))$. From \eqref{eq-putting integer}, we have $p_{(0,0,0),(\{0,1\},(0,0))}=*$ because $d((0,0,0)|_{\{0,1\}},(0,0))$ $=d((0,0),(0,0))=0$. Then $p_{(1,1,0),(\{0,1\},(0,0))}=({\bf e},n_e)$ with ${\bf e}=(0,0,0)$ and $n_e=0$, since $d((1,1,0)|_{\{0,1\}},(0,0))=d((1,1),(0,0))=2$ and it is the first occurrence of ${\bf e}$ in column $(\{0,1\},(0,0))$. It is easy to check that the above array is a $(12,4,3,4)$ PDA.
\end{example}

If $({\bf e},n_{\bf e})$ occurs in $\mathbf{P}$, we say ${\bf e}$ occurs in $\mathbf{P}$. For each vector ${\bf e}$ occurring in $\mathbf{P}$, if the number of its occurrences in each column is at most once, we always omit the occurrence order $n_{\bf e}$. For instance, the PDA in \eqref{eq-PDA-322} can be written as

\begin{eqnarray*}
\begin{footnotesize}
\begin{array}{|c|cccc|cccc|cccc|}\hline
\multicolumn{1}{|c|}{ \multirow{2}*{ \diagbox{${\bf f}$} {$(\mathcal{T},{\bf b} )$}} }&
 \multicolumn{4}{|c|}{ \multirow{1}*{$\mathcal{T}=\{0 ,1\}$}} & \multicolumn{4}{|c|}{ \multirow{1}*{$\mathcal{T}=\{0,2\}$}}& \multicolumn{4}{|c|}{ \multirow{1}*{$\mathcal{T}=\{1,2\}$}}\\ \cline{2-13}
       &00  & 10 & 01 & 11 & 00 & 10&01   &11 & 00  & 10 & 01 &11\\ \hline
000&  *    &*      &*      &110&*      &*      &*      &101&*      &*      &*      &011\\
101&  *    &*      &011&*      &000&*      &*      &*      &*      &110&*      &*\\
011&*      &101&*      &*      &*      &110&*      &*      &000&*      &*      &*\\
110&000&*      &*      &*      &*      &*      &011&*      &*      &*      &101&*\\ \hline
\end{array}\end{footnotesize}
\end{eqnarray*}

Next we will prove that the array $\mathbf{P}$ generated by Construction \ref{construction} is indeed a $(K,F,S)$ PDA. We will show that $\mathbf{P}$ satisfies condition C1 of Definition \ref{def-PDA}.
\begin{proposition}
\label{pro-property-integer}
Let $\mathbf{P}$ be the array generated by Construction \ref{construction}. If there are two distinct entries being the same vector $({\bf e},n_{{\bf e}})$, say $p_{{\bf f},(\mathcal{T},{\bf b})}=p_{{\bf f}',(\mathcal{T}',{\bf b}')}=({\bf e},n_{{\bf e}})$, then the following two statements must hold.
\begin{itemize}
\item [1)] $\mathcal{T}\neq \mathcal{T}'$ and ${\bf f}\neq{\bf f}'$, which implies that they lie in different columns and different rows;
\item [2)] The subarray formed by rows ${\bf f}$, ${\bf f}'$ and columns $(\mathcal{T},{\bf b})$, $(\mathcal{T}',{\bf b}')$ must be of the following form
  \begin{eqnarray*}
    \left(\begin{array}{cc}
      ({\bf e},n_{{\bf e}}) & *\\
      * & ({\bf e},n_{{\bf e}})
    \end{array}\right)~\textrm{or}~
    \left(\begin{array}{cc}
      * & ({\bf e},n_{{\bf e}})\\
      ({\bf e},n_{{\bf e}}) & *
    \end{array}\right).
  \end{eqnarray*}
\end{itemize}
\end{proposition}

\begin{proof}Clearly, each vector entry $({\bf e},n_{{\bf e}})=(e_0,e_1,\ldots,e_{m-1},n_{{\bf e}})$ occurs at most once in each column by \eqref{eq-constr-PDA}, i.e., $(\mathcal{T},{\bf b})\neq (\mathcal{T}',{\bf b}')$. Let
\begin{eqnarray*}
&{\bf f}=(f_0,f_1,\ldots,f_{m-1}),&\  {\bf f}'=(f'_0,f'_1,\ldots,f'_{m-1}),\\
&\mathcal{T}=\{\delta_0,\delta_1,\ldots,\delta_{t-1}\},&\ \mathcal{T}'=\{\delta'_0,\delta'_1,\ldots,\delta'_{t-1}\},\\
&{\bf b}=(b_0,b_1,\ldots,b_{t-1}),&\ {\bf b}'=(b'_0,b'_1,\ldots,b'_{t-1}).
\end{eqnarray*}
If $\mathcal{T}=\mathcal{T}'$, then ${\bf b}\neq {\bf b}'$ holds. Without loss of generality, assume that $b_0\neq b'_0$. By \eqref{eq-putting integer} we have $e_{\delta_0}=b_0=b'_0$, which contradicts to the hypothesis. So we have $\mathcal{T}\neq\mathcal{T}'$, which implies that there must exist two distinct integers, say $i$, $i'\in[0,m)$, satisfying
$$i\in \mathcal{T},\  i\not\in\mathcal{T}'\ \ \ \ \hbox{and}\ \ \ \ i'\in \mathcal{T}',\  i'\not\in\mathcal{T}.$$
Without loss of generality, assume that $i=\delta_0\in \mathcal{T}$ and $i'=\delta'_0\in \mathcal{T}'$. From Construction \ref{construction}, we have
\begin{eqnarray*}
e_{\delta_0}=b_{0}=f'_{\delta_0},\ \ \ \hbox{and} \ \ \ e_{\delta'_0}=b'_0=f_{\delta'_0}.
\end{eqnarray*}
Hence $p_{{\bf f}, (\mathcal{T}',{\bf b}')}=p_{{\bf f}', (\mathcal{T},{\bf b})}=*$  because both $d({\bf f}|_{\mathcal{T}'},{\bf b}')<t$ and $d({\bf f}'|_{\mathcal{T}},{\bf b})<t$ hold. This implies ${\bf f}\neq {\bf f}'$. The proof is complete.
\end{proof}
\begin{remark}
\label{con-PDA}
The array $\mathbf{P}$ generated by Construction \ref{construction} is a $(K, F, S)$ PDA with $K=|\mathcal{K}|$ and $F=|\mathcal{F}|$ by Proposition \ref{pro-property-integer}. Furthermore, if the number of stars in each column is the same, such as $Z$, then $\mathbf{P}$ is a $(K, F, Z, S)$ PDA. It is worth noting that the number of occurrences of each vector $({\bf e},n_{{\bf e}})$ in $\mathbf{P}$, i.e. the coded gain, is at most ${m\choose t}$ by Proposition \ref{pro-property-integer}.
\end{remark}

Under the framework of Construction \ref{construction}, we only need to choose a row index matrix and a column index set appropriately in order to design a coded caching scheme. Next a new scheme by Construction \ref{construction} will be presented, which can be regarded as a generalization of the MN scheme \cite{MN} and the PDA constructed in \cite{SZG}.

\begin{theorem}
\label{th-subsets-PDA}
For any positive integers $m$, $s$, $t$ and $\omega$ with $\omega\leq t\leq s$ and $s+t-2\omega\leq m$, there exists a $({t\choose \omega}{m\choose t},{m\choose s}, {m\choose s}-{m-t\choose s-\omega},{m\choose s+t-2\omega})$ PDA which can realize a $({t\choose \omega}{m\choose t},M,N)$ coded caching scheme with $\frac{M}{N}=1-{m-t\choose s-w}/{m\choose s}$, subpacketization $F={m\choose s}$ and transmission load $R={m\choose s+t-2\omega}/{m\choose s}$.
\end{theorem}
\begin{proof}
Let $q=2$. For  positive integers $m$, $s$, $\omega$ and $t$ with $\omega\leq t\leq s$ and $s+t-2\omega\leq m$, let $\mathcal{F}$ be a collection of all $m$-vectors with weight $s$, i.e.,
\begin{eqnarray}
\label{eq-subset-F}
\mathcal{F}=\{{\bf f}\ |\ wt({\bf f})=s, {\bf f}\in [0,2)^m\}.
\end{eqnarray}
Let
\begin{eqnarray}
\label{eq-T-column-vecetor}
\mathcal{K}=\mathfrak{T} \times \mathcal{B},\ \ \mathfrak{T}={[0,m)\choose t}, \ \ \mathcal{B}=\{{\bf b}\ |\ wt({\bf b})=t-\omega, {\bf b}\in[0,2)^t\}.
\end{eqnarray}
Then $K=|\mathcal{K}|={t\choose \omega}$ and $F=|\mathcal{F}|={m\choose s}$. From Construction \ref{construction}, an ${m\choose s}\times {t\choose \omega}{m\choose t}$ array $\mathbf{P}$ is obtained. For each column $(\mathcal{T},{\bf b})\in \mathcal{K}$, if there exists a row ${\bf f}\in \mathcal{F}$ satisfying that $p_{{\bf f},(\mathcal{T},{\bf b})}$ is a non-star entry, then $wt({\bf f}|_{\mathcal{T}})=w$ and $wt({\bf f}|_{[0,m)\setminus\mathcal{T}})=s-w$ by \eqref{eq-subset-F}, \eqref{eq-T-column-vecetor} and Construction \ref{construction}. Moreover, ${\bf f}|_{\mathcal{T}}$ is determined by $(\mathcal{T},{\bf b})$ since $q=2$. So the number of non-star entries in each column is ${m-t\choose s-\omega}$. Then there are ${m\choose s}-{m-t\choose s-\omega}$ stars in each column, i.e., $Z={m\choose s}-{m-t\choose s-\omega}$ and $\frac{M}{N}=1-{m-t\choose s-\omega}/{m\choose s}$.

For each vector ${\bf e}$ occurring in $\mathbf{P}$, the number of its occurrences in each column is at most once. Otherwise if ${\bf e}$ occurs in column $(\mathcal{T},{\bf e}|_{\mathcal{T}})$ at least twice, say $P_{({\bf f},(\mathcal{T},{\bf e}|_{\mathcal{T}}))}=P_{({\bf f}',(\mathcal{T},{\bf e}|_{\mathcal{T}}))}=({\bf e},n_{{\bf e}})$. Let ${\bf e}=(e_0,e_1,\ldots,e_{m-1})$, ${\bf f}=(f_0$, $f_1$, $\ldots$, $f_{m-1})$ and ${\bf f}'=(f'_0$, $f'_1$, $\ldots$, $f'_{m-1})$, we have
\begin{eqnarray*}
&f_j=f'_{j}=e_{j}\ \ \ \ \ \ \      j\not\in \mathcal{T},\\
&f_{j'}=f'_{j'}\neq e_{j'}\ \ \ \ j'\in \mathcal{T}.
\end{eqnarray*}
This implies ${\bf f}={\bf f}'$, which contradicts to the hypothesis. So the vector $({\bf e}, n_{{\bf e}})$ can be written as ${\bf e}$ and we only need to count the number of different vectors ${\bf e}$ occurring in $\mathbf{P}$. By Construction \ref{construction} and equations \eqref{eq-subset-F}, \eqref{eq-T-column-vecetor}, each vector occurring in $\mathbf{P}$ has weight $s+t-2\omega$. In fact, each vector ${\bf e}$ with wight $s+t-2\omega$ must occur in $\mathbf{P}$.
Without loss of generality, let us consider the vector ${\bf e}=(e_0,e_1,\ldots, e_{m-1})$ such that $e_j=1$ for $j\in[0,s+t-2\omega)$ and $e_j=0$ for $j\in[s+t-2\omega,m)$. Clearly $wt({\bf e})=s+t-2\omega$. Then the vector ${\bf f}=(f_0,f_1,\ldots,f_{m-1})$ with
\begin{eqnarray*}
f_j=\left\{\begin{array}{cc}
1 & \text{if}\ j\in [0,s-\omega)\cup [s+t-2\omega, s+t-\omega) \\
0 & \text{otherwise}
\end{array}\right.
\end{eqnarray*}
belongs to $\mathcal{F}$ by \eqref{eq-subset-F} and the element $(\mathcal{T},{\bf b})$ with $\mathcal{T}=[s-\omega,s-\omega+t-1)$, ${\bf b}=(b_0,b_1,\ldots,b_{t-1})$ and
\begin{eqnarray*}
b_j=\left\{\begin{array}{cc}
1 & \text{if}\ j\in[0,t-\omega) \\
0 & \text{otherwise}
\end{array}\right.
\end{eqnarray*}
belongs to $\mathcal{K}$ by \eqref{eq-T-column-vecetor} since $wt({\bf f})=s$,
$|\mathcal{T}|=t$ and $wt({\bf b})=t-\omega$ hold. Furthermore, we have $p_{{\bf f},(\mathcal{T},{\bf b})}={\bf e}$ by \eqref{eq-constr-PDA}. As a result, there are ${m\choose s+t-2\omega}$ different vectors occurring in $\mathbf{P}$. So we have $S={m\choose s+t-2\omega}$ and the transmission load is $R=\frac{S}{F}={m\choose s+t-2\omega}/{m\choose s}$. The proof is complete.
\end{proof}

\begin{example}
When $m=4$, $s=t=2$ and $\omega=1$, from Theorem \ref{th-subsets-PDA} we have $K={t\choose \omega}{m\choose t}=12$, $F={m\choose s}=6$, $Z={m\choose s}-{m-t\choose s-\omega}=4$ and $S={m\choose s+t-2\omega}=6$. The following array, which is generated by Construction \ref{construction} using $\mathcal{F}$ in  \eqref{eq-subset-F} and $\mathcal{K}$ in \eqref{eq-T-column-vecetor}, is a $(12,6,4,6)$ PDA. Here we omit the occurrence order $n_{\bf e}$ since each vector ${\bf e}$ occurs at most once in each column.
\begin{eqnarray*}
\label{eq-PDA-421}\begin{footnotesize}
\begin{array}{|c|c|c|c|c|c|c|c|c|c|c|c|c|}\hline
\multicolumn{1}{|c|}{ \multirow{2}*{ \diagbox{${\bf f}$} {$(\mathcal{T},{\bf b} )$}} }&
 \multicolumn{2}{|c|}{ \multirow{1}*{$\mathcal{T}=\{0,1\}$ }}&  \multicolumn{2}{|c|}{ \multirow{1}*{$\mathcal{T}=\{0,2\}$ }}&  \multicolumn{2}{|c|}{ \multirow{1}*{$\mathcal{T}=\{1,2\}$ }}&  \multicolumn{2}{|c|}{ \multirow{1}*{$\mathcal{T}=\{0,3\}$ }}
 &  \multicolumn{2}{|c|}{ \multirow{1}*{$\mathcal{T}=\{1,3\}$ }}&  \multicolumn{2}{|c|}{ \multirow{1}*{$\mathcal{T}=\{2,3\}$ }}\\ \cline{2-13}
 &10 &01  &10  &01  &10  &01  &10  &01  &10  &01  &10  &01\\ \hline
1100&*   &*   &*   &0110&*   &1010&*   &0101&*   &1001&*   &*  \\ \hline
1010&*   &0110&*   &*   &1100&*   &*   &0011&*   &*   &*   &1001 \\ \hline
1001&*   &0101&*   &0011&*   &*   &*   &*   &1100&*   &1010&*  \\ \hline
0110&1010&*   &1100&*   &*   &*   &*   &*   &*   &0011&*   &0101 \\ \hline
0101&1001&*   &*   &*   &*   &0011&1100&*   &*   &*   &0110&* \\ \hline
0011&*   &*   &1001&*   &0101&*   &1010&*   &0110&*   &*   &*  \\ \hline
\end{array}\end{footnotesize}
\end{eqnarray*}

\end{example}
\begin{remark}
\label{re-advan-subset}
\begin{itemize}
\item When $\omega=0$, from Theorem \ref{th-subsets-PDA}, we have a $({m\choose t},{m\choose s}, {m\choose s}-{m-t\choose s},{m\choose s+t})$ PDA, which is exactly the first PDA constructed in \cite{SZG}.
     When $t=1$, $K=m$, $\omega=0$ and $s=KM/N$, from Theorem \ref{th-subsets-PDA} we have a $(K,{K\choose KM/N}, {K\choose KM/N}-{K-1\choose KM/N-1},{K\choose KM/N+1})$ PDA, which realizes the MN scheme in \cite{MN}.
\item When $\omega\neq 0$, the scheme from Theorem \ref{th-subsets-PDA} may have better performance than that of the scheme when $\omega= 0$. For example, in Table \ref{table-comparison}, compared with Scheme 2 with $\omega=0$, Scheme 1 with $\omega=2$ has more users, less memory size and subpacketization but has larger coded gain; Compared with Scheme 3 with $\omega=0$, Scheme 4 with $\omega=1$ has more users, the same memory size and subpacketization but has larger coded gain.
\end{itemize}
\end{remark}
\begin{table}[!htp]
\centering
  \caption{Comparisons of the schemes with $\omega=0$ and $\omega>0$}\label{table-comparison}
\begin{tabular}{|c|c|c|c|c|c|c|c|} \hline
Schemes& $K$	&	$\frac{M}{N}$	&	$F$	&	$R$	&	coded gain 	&	$\omega$		\\	\hline
Scheme 1& 360	&	0.9	&	210	&	0.57143	&	63	&	2		\\	\hline	
Scheme 2&120	&	0.91667	&	252	&	0.17857	&	56	&	0		\\	\hline
Scheme 3&120	&	0.83333 &	210	&	0.57143	&	35	&	0		\\	\hline	
Scheme 4&360	&	0.83333 &	210	&	1.2	&	50	&	1		\\	\hline	
\end{tabular}
\end{table}

\section{Necessary conditions for the row index matrix}
\label{sect-construction via OA}
In this section,  we consider the case that the column index set is  $\mathcal{K}={[0,m)\choose t}\times [0,q)^t$  which implies $K={m\choose t}q^t$. The following definition will be used in the rest of this paper.
\begin{definition}[\cite{Stinson}]
\label{defi-OA}
Let $\mathbf{F}$ be an $F\times m$ matrix over $[0,q)$ for positive integers $F$, $m$,  $q\geq 2$, and $s\leq m$.
\begin{itemize}
\item $\mathbf{F}$ is an orthogonal array (OA) with strength $s$, denoted by OA$_{\lambda}(F,m,q,s)$, if every $1\times s$ row vector appears exactly $\lambda$ times in $\mathbf{F}|_{\mathcal{S}}$ for each $\mathcal{S}\in {[0,m)\choose s}$.
\item $\mathbf{F}$ is a covering array (CA) with strength $s$, denoted by CA$_{\lambda}(F,m,q,s)$, if every $1\times s$ row vector appears at least $\lambda$ times in $\mathbf{F}|_{\mathcal{S}}$ for each $\mathcal{S}\in {[0,m)\choose s}$.
\end{itemize}
\end{definition}
It is well known that $F=\lambda q^s$  for any OA$_{\lambda}(F,m,q,s)$  and thus OA$_{\lambda}(F,m,q,s)$ is sometimes written as OA$_{\lambda}(m,q,s)$ for short \cite{Stinson}. The parameter $\lambda$ is the index of the orthogonal array. If $\lambda$ is omitted, then it is understood to be $1$. Similarly,  the index of the covering array can be omitted when it is $1$. Clearly for any CA$(F,m,q,s)$, we must have $F\geq q^s$ because any OA$(m,q,s)$ is also a CA$(q^s,m,q,s)$.

When $m=3$, $q=3$ and $s=2$, let us consider the array in \eqref{eq-OA-1}. For each $\mathcal{S}=\{0,1\}$, $\{0,2\}$ and $\{1,2\}$, we have \begin{eqnarray*}
\mathbf{F}|_{\{0,1\}}=\left(
\begin{array}{ccc}
0&0\\
1&0\\
0&1\\
1&1
\end{array}\right)\ \ \ \
\mathbf{F}|_{\{0,2\}}=\left(
\begin{array}{ccc}
0&0\\
1&1\\
0&1\\
1&0
\end{array}\right)\ \ \ \
\mathbf{F}|_{\{ 1, 2\}}=\left(
\begin{array}{ccc}
0&0\\
0&1\\
1&1\\
1&0
\end{array}\right)
\end{eqnarray*}
Clearly $\mathbf{F}$ in \eqref{eq-OA-1} satisfies Definition \ref{defi-OA}. So $\mathbf{F}$ is an OA$(3,2,2)$ and it is also a CA$(4,3,2,2)$.

Next the necessary condition for the row index matrix will be discussed when each user has the same memory size or the coded gain is the largest.

\subsection{Necessary conditions for the row index matrix when each user has the same memory size}
\label{sub-C2}
\begin{theorem}\label{pro-property-OA}
For any positive integers $q$, $m$ and $t$ with $t\leq m$, let $K={m \choose t}q^t$. Then the array $\mathbf{P}$ generated by Construction \ref{construction} is a $(K,F,Z,S)$ PDA if and only if the row index matrix $\mathbf{F}$ is an OA$_{\lambda}(m,q,t)$ with $\lambda=\frac{F-Z}{(q-1)^t}$, i.e.,
every $1\times t$ row vector appears exactly $\lambda$ times in $\mathbf{F}|_{\mathcal{T}}$ for any $t$-subset $\mathcal{T}\subseteq [0,m)$.
\end{theorem}
\begin{proof}
First we assume that $\mathbf{P}$ is a $(K,F,Z,S)$ PDA with $K={m \choose t}q^t$. By Construction \ref{construction},  the column index set is $\mathcal{K}={[0,m)\choose t}\times [0,q)^t$ and the row index matrix is an $F\times m$ matrix $\mathbf{F}$.
For each $\mathcal{T}\in {[0,m)\choose t}$ and for each ${\bf b}\in[0,q)^t$,  let  $h_{\mathcal{T},{\bf b}}$ denote the number of occurrences of ${\bf b}$ in $\mathbf{F}|_{\mathcal{T}}$.

When $t=1$, for each $\mathcal{T}\in {[0,m)\choose t}$, we build a binary table  (see Table \ref{tablet1}), where the entry at row $\{\mathcal{T}, {\bf b}\}$ and column $\mathbf{f}|_{\mathcal{T}}$ is $0$ if $p_{\mathbf{f}, (\mathcal{T}, {\bf b})} = *$, otherwise it is $1$. According to Table \ref{tablet1}, a $q\times q$ binary matrix $\Phi_1$ is defined in \eqref{Phi1}.
\begin{table}[h]
  \centering
  \caption{The binary table for $t=1$ \label{tablet1}}
  \begin{tabular}{|c|c|c|c|c|}
\hline
\diagbox{$\{\mathcal{T}, {\bf b}\}$}{${\bf f}|_{\mathcal{T}}$}
  &  $(0)$   &    $(1)$ &   $\cdots$           & $(q-1)$   \\
\hline
$(\mathcal{T},(0))$ &$0$ &  $1$ &   $\cdots$ & $1$ \\
$(\mathcal{T},(1))$ &$1$ &  $0$ &   $\cdots$ & $1$ \\
$\vdots$       & $\vdots$&$\vdots$&$\ddots$&$\vdots$  \\
$(\mathcal{T},(q-1))$ &$1$ &  $1$ &   $\cdots$ & $0$ \\
\hline
   \end{tabular}
\end{table}
\begin{eqnarray}
\label{Phi1}
\Phi_1=\left(\begin{array}{cccc}
0      &     1    &   \cdots    & 1 \\
1      &     0    &   \cdots    & 1 \\
\vdots &   \vdots &   \ddots    & \vdots \\
 1     &    1     &  \cdots     &   0
\end{array}
\right)
\end{eqnarray}

When $t>1$, Table \ref{tabletn+1} can be obtained similarly, where the row and column labels are arranged in lexicographic order. Then a $q^t\times q^t$ binary matrix $\Phi_t$ is recursively defined in \eqref{Phit}.
\begin{table}[http!]
  \centering
  \caption{The binary table for any $t>1$ \label{tabletn+1}}
  \begin{tabular}{|c|c |c|c|c|}
\hline
\diagbox{$\{\mathcal{T}, {\bf b}\}$}{${\bf f}|_{\mathcal{T}}$}
& \tabincell{l}{$(0,b_1,\ldots,b_{t-1})$,  \\ $b_h\in [0,q)$, $h\in[1,t)$}&
 \tabincell{l}{$(1,b_1,\ldots,b_{t-1})$,\\ $b_h\in [0,q)$, $h\in[1,t)$} &$\cdots$&
 \tabincell{l}{$(q-1,b_1,\ldots,b_{t-1})$,\\ $b_h\in [0,q)$, $h\in[1,t)$}   \\
\hline
 \tabincell{l}{$(\mathcal{T},(0,b_1,\ldots,b_{t-1})),$\\ $b_h\in [0,q)$, $h\in[1,t)$}
 & { $0$} & { $\Phi_{t-1}$} & $\cdots$ & { $\Phi_{t-1}$}\\%
\hline
 \tabincell{l}{$(\mathcal{T},(1,b_1,\ldots,b_{t-1})),$\\ $b_h\in [0,q)$, $h\in[1,t)$}
 & { $\Phi_{t-1}$} & { $0$} & $\cdots$ & { $\Phi_{t-1}$}\\  \hline
 $\vdots$&$\vdots$&$\vdots$&$\ddots$&$\vdots$\\ \hline
 \tabincell{l}{$(\mathcal{T},(q-1,b_1,\ldots,b_{t-1})),$\\ $b_h\in [0,q)$, $h\in[1,t)$}
 & { $\Phi_{t-1}$} & { $\Phi_{t-1}$}& $\cdots$ &{ $0$}\\ \hline
   \end{tabular}
\end{table}

\begin{eqnarray}
\label{Phit}
\Phi_{t}=\left(\begin{array}{cccc}
\Large {0}      &     \Large {\Phi_{t-1}}    &   \cdots    & \Large {\Phi_{t-1}} \\
\Large {\Phi_{t-1}}      &    \Large{ 0}    &   \cdots    & \Large {\Phi_{t-1}}\\
\vdots &   \vdots &   \ddots    & \vdots \\
\Large {\Phi_{t-1}}     &   \Large {\Phi_{t-1}}    &  \cdots     &  \Large{ 0}
\end{array}
\right)
\end{eqnarray}

Let $${\bf h}_{\mathcal{T},t}=(h_{\mathcal{T},(0,0,\ldots,0)},h_{\mathcal{T},(0,0,\ldots,1)},\ldots,h_{\mathcal{T},(q-1,q-1,\ldots,q-1)})^{\top}.$$
Since the number of non-star entries in each column of $\mathbf{P}$ is $F-Z$, from Table \ref{tabletn+1},  we have
\begin{equation}
\label{et1}
 \Phi_t{\bf h}_{\mathcal{T},t}=(F-Z)\cdot {\bf 1}_{q^t\times 1},
\end{equation}
where ${\bf 1}_{q^t\times 1}$ is the $q^t\times 1$ vector with all entries equal to $1$. Obviously, $\Phi_1$ in \eqref{Phi1} is a circulant matrix and it is easy to check that it is invertible. Using elementary row operations the determinant of  $\Phi_t$ can be obtained as $|\Phi_t|=(q-1)(-1)^{q-1}|\Phi_{t-1}|^q$. Hence $\Phi_t$ is also invertible for $t >1$. Furthermore, the number of $1$ in each row of $\Phi_t$ is $(q-1)^t$. Consequently the unique solution of \eqref{et1} is $$h_{\mathcal{T},(0,0,\ldots,0)}=h_{\mathcal{T},(0,0,\ldots,1)}=\ldots=h_{\mathcal{T},(q-1,q-1,\ldots,q-1)}=\frac{F-Z}{(q-1)^t}.$$ So the row index matrix $\mathbf{F}$ is an OA$_{\lambda}(m,q,t)$ with $\lambda=\frac{F-Z}{(q-1)^t}$.

Conversely,  we assume that the row index matrix $\mathbf{F}$ is an OA$_{\lambda}(m,q,t)$. Then we have $F=\lambda q^t$ and there are $\lambda (q-1)^t$ non-star entries in each column of $\mathbf{P}$ by Construction \ref{construction}. This implies that there are $\lambda q^t-\lambda (q-1)^t$ stars in each column, i.e. $Z=\lambda q^t-\lambda (q-1)^t$. So $\mathbf{P}$ is a $(K,F,Z,S)$ PDA with $K={m \choose t}q^t$ and $\frac{Z}{F}=1-(\frac{q-1}{q})^t$ by Proposition \ref{pro-property-integer}. Hence the proof is complete.
\end{proof}
From the above proof, the following result is obvious.
\begin{remark}
\label{remark-caching-OA}
If the row index matrix $\mathbf{F}$ is an OA$_{\lambda}(m,q,t)$ and the column index set $\mathcal{K}={[0,m)\choose t}\times [0,q)^t$, the array $\mathbf{P}$ generated by Construction \ref{construction} is a $(K,F,Z,S)$ PDA with $K={m \choose t}q^t$ and $\frac{Z}{F}=1-(\frac{q-1}{q})^t$.
\end{remark}

\subsection{Necessary conditions for the row index matrix when the coded gain is the maximum}
\label{sub-lagest-gain}
From Remark \ref{con-PDA}, the number of  occurrences  (i.e. the coded gain) of each vector in the array $\mathbf{P}$ is at most ${m\choose t}$, i.e., the largest coded gain is ${m\choose t}$. In practice, for the fixed number of users, subpacketization and memory ratio, we prefer the coded gain as large as possible.
\begin{theorem}
\label{pro-property-gain}

For any positive integers $q$, $m$ and $t$ with $t\leq m$, let the array $\mathbf{P}$ generated by Construction \ref{construction} be a $(K,F,Z,S)$ PDA with $K={m \choose t}q^t$. If each vector $({\bf e},n_{{\bf e}})$ in $\mathbf{P}$ occurs exactly ${m\choose t}$ times, i.e. the coded gain is ${m \choose t}$, then we have
\begin{itemize}
\item [1)]the row index matrix $\mathbf{F}$ is a CA$(F,m,q,m-t)$, i.e., every $1\times (m-t)$ row vector appears at least once in $\mathbf{F}|_{\mathcal{S}}$ for any $(m-t)$-subset $\mathcal{S}\subseteq [0,m)$;
\item [2)]the subpacketization $F\geq q^{m-t}$.
\end{itemize}
\end{theorem}
\begin{proof}
Since the array $\mathbf{P}$ generated by Construction \ref{construction} is a $(K,F,Z,S)$ PDA with $K={m \choose t}q^t$, we have the column index set $\mathcal{K}={[0,m)\choose t}\times [0,q)^t$ by Construction \ref{construction}. Moreover, the row index matrix $\mathbf{F}$ is an OA$_{\lambda}(m,q,t)$ by Theorem \ref{pro-property-OA}.

When $t\geq m-t$, the result holds because  an OA$_{\lambda}(m,q,t)$ is also an OA$_{\lambda'}(m,q,m-t)$ with $\lambda'=\lambda q^{2t-m}$. When $t<m-t$, without loss of generality, let $\overline{\mathcal{T}}=\{0,1,\ldots,m-t-1\}$, we need to check that any $(m-t)$-dimensional row vector,  without loss of generality,  say  $(0,0,\ldots,0)$,  is a row vector of $\mathbf{F}|_{\overline{\mathcal{T}}}$.

Since $\mathbf{F}$ is an OA$_{\lambda}(m,q,t)$, there exists  a row vector ${\bf f}_0=(f_{0,0},f_{0,1},\ldots,f_{0,m-1})=(0,\ldots,0,f_{0,t},\cdots,f_{0,m-1})$ in $\mathcal{F}$. Let $i_0$ be the smallest coordinate satisfying $f_{0,i_0}\neq 0$. If $i_0\geq m-t$, then the result holds. Otherwise if $i_0< m-t$, then we choose $\mathcal{T}_0=\{i_0,m-t+1,\ldots,m-1\}$ and ${\bf b}_0=(0,b_{0,1},\ldots,b_{0,t-1})$ such that $b_{0,j}\neq f_{0,m-t+j}$ for $1\leq j<t$. By Construction \ref{construction}, there is a vector ${\bf e}_0=(e_{0,0},e_{0,1},\ldots,e_{0,m-1})=(0,\ldots,0,f_{0,i_0+1},\ldots,f_{0,m-t},b_{0,1},\ldots,b_{0,t-1})$ occurring at row ${\bf f}_0$ and column $(\mathcal{T}_0,{\bf b}_0)$.

Since the coded gain is ${m\choose t}$, each vector occurring in $\mathbf{P}$ must appear in exactly ${m\choose t}$ columns. Let $\mathcal{T}_1=\{m-t,m-t+1,\ldots,m-1\}$. Then ${\bf e}_0$ must appear at column $(\mathcal{T}_1,{\bf e}_0|_{\mathcal{T}_1})$ and row
$${\bf f}_1=(f_{1,0}, f_{1,1}, \ldots, f_{1,m-1})=(0,\ldots,0,e_{0,i_0+1},\ldots,e_{0,m-t-1},f_{1,m-t},\ldots,f_{1,m-1})$$
with $f_{1,j}\neq e_{0,j}$ for $j\in[m-t,m)$. Let $i_1$ be the smallest coordinate satisfying $f_{1,i_1}\neq 0$, i.e., ${\bf f}_1=(0,\ldots,0,f_{1,i_1},\ldots,f_{1,m-1})$, clearly $i_1>i_0$. If $i_1\geq m-t$, the result holds. Otherwise if $i_1< m-t$, we can repeat the same process until we obtain a row ${\bf f}_u=(0,\ldots,0,f_{u,i_u},\ldots,f_{u,m-1})$ with $f_{u,i_u}\neq 0$ and $i_u\geq m-t$.  This shows that $(0, 0, \ldots, 0)$ is a row vector of   $\mathbf{F}|_{\overline{\mathcal{T}}}$.  Hence the row index matrix $\mathbf{F}$ is a CA$(F,m,q,m-t)$, which implies $F\geq q^{m-t}$.
The proof is complete.
\end{proof}

In summary, the following result is obtained.
\begin{theorem}
\label{lowerbounds}
For any positive integers $q$, $m$ and $t$ with $t\leq m$, let the array $\mathbf{P}$ generated by Construction \ref{construction} be a $(K,F,Z,S)$ PDA with $K={m \choose t}q^t$. The coded caching scheme realized by $\mathbf{P}$ has memory ratio $\frac{Z}{F}=1-(\frac{q-1}{q})^t$ and transmission load $R\geq (q-1)^t$. Furthermore, if the transmission load achieves the lower bound, i.e. $R= (q-1)^t$, then the subpacketization $F\geq q^{m-t}$.
\end{theorem}
\begin{proof}
From Theorem \ref{pro-property-OA} and Remark \ref{remark-caching-OA}, the coded caching scheme realized by $\mathbf{P}$ has the memory ratio $\frac{Z}{F}=1-(\frac{q-1}{q})^t$. From Remark \ref{con-PDA}, each vector $({\bf e},n_{{\bf e}})$ in $\mathbf{P}$ occurs at most ${m\choose t}$ times, which implies that the coded gain is at most ${m\choose t}$. Hence $S\geq \frac{K(F-Z)}{{m\choose t}}=F(q-1)^t$. Consequently we have $R=\frac{S}{F}\geq (q-1)^t$. If  the transmission load achieves the lower bound, i.e.,  $R= (q-1)^t$, then  the coded gain is ${m\choose t}$. Hence we have $F\geq q^{m-t}$ from Theorem \ref{pro-property-gain}.
\end{proof}

\section{New schemes via orthogonal arrays}
\label{sec-OAs}
In this section,  we demonstrate our framework by presenting  two more explicit new schemes that use different OAs as the row index matrix $\mathbf{F}$ in  Construction \ref{construction}.
These new schemes have significantly advantages in subpacketization compared with the schemes in \cite{SZG}.
Again, the column index set $\mathcal{K}$ is  the set ${[0,m)\choose t}\times [0,q)^t$.
\subsection{New scheme via OA$(m,q,m-1)$}
In this subsection, we use a trivial orthogonal array OA$(m,q,m-1)$ as the row index matrix $\mathbf{F}$. Then the array $\mathbf{P}$
generated by Construction \ref{construction} has the following properties.
\begin{lemma}
\label{le-occurrence number-t}
  For any positive integers $m$, $t$ with $t\leq m$ and $q\geq 2$, let the column index set $\mathcal{K}$ be  ${[0,m)\choose t}\times [0,q)^t$ and the row index matrix $\mathbf{F}$ be an OA$(m,q,m-1)$. For each vector $({\bf e}, n_{{\bf e}})$ in the array $\mathbf{P}$ generated by Construction \ref{construction}, we have
\begin{itemize}
\item [1)] ${\bf e}$ occurs in exactly ${m\choose t}$ columns;
\item [2)]  The number of times in each column that ${\bf e}$ occurs is the same.
\end{itemize}
\end{lemma}
\begin{proof}
Recall that if $({\bf e},n_{\bf e})$ occurs in $\mathbf{P}$, we say ${\bf e}$ occurs in $\mathbf{P}$.
For any vector ${\bf e}=(e_0,\ldots,e_{m-1})$ occurring in $\mathbf{P}$,  we assume that $p_{{\bf f},(\mathcal{T},{\bf e}|_{\mathcal{T}})}=({\bf e},n_{{\bf e}})$ with ${\bf f}=(f_0,\ldots,f_{m-1})$. For  any fixed column $(\mathcal{T},{\bf e}|_{\mathcal{T}})$, the collection of rows in which ${\bf e}$ occurs is denoted by $\mathcal{F}_{\mathcal{T},{\bf e}}$, i.e.,
\begin{eqnarray*}
\mathcal{F}_{\mathcal{T},{\bf e}}=\{{\bf f}\in \mathcal{F}\ |\ f_i=e_i \ \textrm{for} \ i\in [0,m)\setminus \mathcal{T} \ \textrm{and} \ f_i\neq e_i \ \textrm{for} \ i\in \mathcal{T}\}
\end{eqnarray*}
For any $\mathcal{T}'\in {[0,m) \choose t}$ with $|\mathcal{T}\cap\mathcal{T}'|=t-1$, assume that $\mathcal{T}\setminus \mathcal{T}'=\{w\}$ and $\mathcal{T}'\setminus \mathcal{T}=\{w'\}$. Clearly, we have $w\neq w'$. For the  fixed column $(\mathcal{T}',{\bf e}|_{\mathcal{T}'})$, the collection of rows in which ${\bf e}$ occurs is
 \begin{eqnarray*}
\mathcal{F}_{\mathcal{T}',{\bf e}}=\{{\bf f}\in \mathcal{F}\ |\ f_i=e_i \ \textrm{for} \ i\in [0,m)\setminus \mathcal{T}' \ \textrm{and} \ f_i\neq e_i \ \textrm{for} \ i\in \mathcal{T}'\}
\end{eqnarray*}
Since $\mathbf{F}$ is an OA$(m,q,m-1)$, we can define a mapping $\psi$ from $\mathcal{F}_{\mathcal{T},{\bf e}}$ to $\mathcal{F}_{\mathcal{T}',{\bf e}}$ as follows
\begin{eqnarray*}
\psi({\bf f})= {\bf f}',\ \ \ \ \ \ \ \  {\bf f}\in \mathcal{F}_{\mathcal{T},{\bf e}},\ \ {\bf f}'\in\mathcal{F}_{\mathcal{T}',{\bf e}},
\end{eqnarray*}
where $f'_i=f_i$ for $i\in [0,m)\setminus \{w,w'\}$ and $f'_w=e_w$. Since every row vector in $[0,q)^{m-1}$ occurs exactly once in submatrix $\mathbf{F}|_{[0,m)\setminus\{w'\}}$, there is only one row vector ${\bf f}'\in\mathcal{F}$ corresponding to ${\bf f}$. Furthermore, we can show that $f'_{w'}\neq e_{w'}$, which implies ${\bf f}'\in \mathcal{F}_{\mathcal{T}',{\bf e}}$. Otherwise if $f'_{w'}= e_{w'}$,  then the $(m-1)$-dimensional vector ${\bf f}|_{[0,m)\setminus \{w\}}$ occurs at least twice in $\mathbf{F}|_{[0,m)\setminus \{w\}}$, which is contradictory to that $\mathbf{F}$ is an OA$(m,q,m-1)$. So $\psi$ is an injective mapping from $\mathcal{F}_{\mathcal{T},{\bf e}}$ to $\mathcal{F}_{\mathcal{T}',{\bf e}}$. Then we have $|\mathcal{F}_{\mathcal{T},{\bf e}}|\leq |\mathcal{F}_{\mathcal{T}',{\bf e}}|$. Similar to the above mapping $\psi$,  we can also define another injective mapping from $\mathcal{F}_{\mathcal{T}',{\bf e}}$ to $\mathcal{F}_{\mathcal{T},{\bf e}}$ as follows
\begin{eqnarray*}
\phi({\bf f}')= {\bf f},\ \ \ \ \ \ \ \  {\bf f}'\in \mathcal{F}_{\mathcal{T}',{\bf e}},\ \ {\bf f}\in\mathcal{F}_{\mathcal{T},{\bf e}},
\end{eqnarray*}
where $f_i=f'_i$ for $i\in [0,m)\setminus \{w,w'\}$ and $f_{w'}=e_{w'}$.  Similarly, we have $|\mathcal{F}_{\mathcal{T},{\bf e}}|\geq |\mathcal{F}_{\mathcal{T}',{\bf e}}|$.  Therefore $|\mathcal{F}_{\mathcal{T},{\bf e}}|=|\mathcal{F}_{\mathcal{T}',{\bf e}}|$, which implies that ${\bf e}$ occurs the same number of times in columns $(\mathcal{T},{\bf e}|_{\mathcal{T}})$ and $(\mathcal{T}',{\bf e}|_{\mathcal{T}'})$, respectively.
For any $\mathcal{T}''\in {[0,m)\choose t}$,  we can find a sequence $\mathcal{T}_0, \ldots, \mathcal{T}_k\in {[0,m)\choose t}$ such that $\mathcal{T}_0 = \mathcal{T}$,  $\mathcal{T}_k = \mathcal{T}''$, and $|\mathcal{T}_i\cap\mathcal{T}_{i+1}|=t-1$ for $0\leq i \leq k-1$. Hence ${\bf e}$ occurs the same number of times in columns $(\mathcal{T},{\bf e}|_{\mathcal{T}})$ and $(\mathcal{T}'',{\bf e}|_{\mathcal{T}''})$ respectively.  The proof is complete.
\end{proof}
From Lemma \ref{le-occurrence number-t}, a new scheme can be obtained as follows.
\begin{theorem}
\label{th-s=m-1}
For any positive integers $m$, $t$ with $t< m$,  and $q\geq 2$, there always exists an $({m\choose t}q^t,q^{m-1}, q^{m-1}-(q-1)^tq^{m-t-1}, (q-1)^tq^{m-1})$ PDA, which can realize a $({m\choose t}q^t,M,N)$ coded caching scheme with $\frac{M}{N}=1-(\frac{q-1}{q})^t$, subpacketization $F=q^{m-1}$ and transmission load $R=(q-1)^t$.
\end{theorem}
\begin{proof}
Let the column index set $\mathcal{K}={[0,m)\choose t}\times [0,q)^t$ and
\begin{eqnarray}
\label{eq-generator-s=m-1}
\mathcal{F}=\{(f_{0},f_{1},\ldots,f_{m-2}, \sum_{i=0}^{m-2} f_i)\ |\  f_0,f_1,\ldots,f_{m-2}\in \mathbb{Z}_q\}.
\end{eqnarray}
Then  $F=q^{m-1}$. It is easy to check that the corresponding row index matrix $\mathbf{F}$, whose collection of row vectors is $\mathcal{F}$, is an OA$(m,q,m-1)$, which is also an OA$_{\lambda}(m,q,t)$ with $\lambda=q^{m-1-t}$. By Remark \ref{remark-caching-OA},  the array $\mathbf{P}$ generated by Construction \ref{construction} is a $(K,F,Z,S)$ PDA with memory ratio $\frac{M}{N}=\frac{Z}{F}=1-(\frac{q-1}{q})^t$. Furthermore, by Lemma \ref{le-occurrence number-t}, each vector $({\bf e}, n_{{\bf e}})$ in $\mathbf{P}$ occurs exactly ${m \choose t}$ times. Hence the coded gain is ${m \choose t}$. Consequently we obtain $S=\frac{K(F-Z)}{{m \choose t}}=F(q-1)^t$ and thus the transmission load is $R=\frac{S}{F}=(q-1)^t$.
\end{proof}
\begin{remark}
\label{rem-2}
\begin{itemize}
\item When $t=1$, the scheme in Theorem \ref{th-s=m-1} is exactly the one from \cite{YCTC}, which is listed in the third row of Table \ref{tab-known}. Moreover, it achieves both the lower bounds on transmission load $R$ and subpacketization $F$ in Theorem \ref{lowerbounds}.
\item  For any positive $1< t< m$, the scheme in Theorem \ref{th-s=m-1} has the same number of users $K={m\choose t}q^t$, memory ratio $\frac{M}{N}=1-(\frac{q-1}{q})^t$ and transmission load $R=(q-1)^t$ as the scheme from \cite{SZG}, which is listed in the fourth row of Table \ref{tab-known},  while its subpacketization is just $\frac{1}{q}$ times as that of the scheme from \cite{SZG}.
\item Let the column index set $\mathcal{K}={[0,m)\choose t}\times [0,q)^t$ and $\mathcal{F}=[0,q)^m$. It is easy to check that the array generated by Construction \ref{construction} is a $({m\choose t}q^t,q^m, q^m-(q-1)^tq^{m-t}, (q-1)^tq^m)$ PDA, which is exactly the second PDA in \cite{SZG}.
\end{itemize}
\end{remark}
\begin{remark}
Given a PDA, clearly its subarray is also a PDA. However it is worth noting that it is not easy to find out a rule of deletion such that the resulting subarray of the PDA in \cite{SZG} has the same performance as the PDA in Theorem \ref{th-s=m-1}. In fact the performance of a PDA depends heavily on how the entries are defined. For example, when $t=1$, the scheme in \cite{YCTC} has much smaller subpacketization than that of the scheme in \cite{SZG} for the same number of users, memory size and transmission load. Indeed, the construction of the scheme with $K=mq$ in \cite{YCTC} has to use two formulas, i.e., the first is for the users in $[0,(m-1)q)$ (i.e., formula (37) in \cite{YCTC}) and the second is for the users in $[(m-1)q,mq)$ (i.e., the formula (38) in \cite{YCTC}). In contrast,  our construction provides a simple rule to define these entries, i.e., \eqref{eq-putting integer} and \eqref{eq-constr-PDA} in Construction \ref{construction}.
\end{remark}
\subsection{New scheme via OA$(m,q,m-t)$}
\label{subsec-OA-MDS}
For any positive integers $m$ and $s$ with $s\leq m$ and for any prime power $q$, the following results will used in this subsection.
\begin{itemize}
\item Let $q$ be a prime power and  $\mathbb{F}^m_q$ denote the vector space of all $m$-tuples over the finite field $\mathbb{F}_q$. A vector set $\mathcal{C}$ is called an $[m,s]_q$ linear code over $\mathbb{F}_q$ if $\mathcal{C}$ is a $s$-dimensional subspace of $\mathbb{F}^m_q$. Each vector in $\mathcal{C}$ is called a codeword. An $[m,s]_q$ linear code $\mathcal{C}$ has $q^s$ codewords.
\item The minimum distance of a linear code is the smallest distance between distinct codewords. An $[m,s,d]_q$ code is an $[m,s]_q$ linear code $\mathcal{C}$ with minimum distance $d$. If $q$ is omitted, then it is understood to be $2$. An $[m,s,d]_q$ code is called an $[m,s]_q$ maximum distance separable (MDS) code if $d=m-s+1$. For each codeword ${\bf c}\in \mathcal{C}$, $\mathcal{E}_{\mathcal{C}}({\bf c},m-s)$ is the sphere with radius $m-s$ and centre ${\bf c}$, i.e., $\mathcal{E}_{\mathcal{C}}({\bf c},m-s)=\{{\bf e} \in\mathbb{F}^m_q \ |\ d({\bf e},{\bf c})\leq m-s\}$.
\end{itemize}

\begin{lemma}
\label{lem-sphere}
For any positive integers $m$ and $s$ with $s \leq m$ and for any prime power $q$, let $\mathcal{C}$ be an $[m,s]_q$ MDS code, then $$\bigcup_{{\bf c}\in \mathcal{C}}\mathcal{E}_{\mathcal{C}}({\bf c},m-s)=\mathbb{F}^m_q.$$
\end{lemma}

A key property of an $[m,s]_q$ MDS code is that each codeword in $\mathcal{C}$ is determined by its any $s$ coordinates \cite{Lint}. Based on this property, given an $[m,s]_q$ MDS code $\mathcal{C}$, an OA$(m,q,s)$ can be obtained by taking the codewords in $\mathcal{C}$ as its row vectors \cite{Stinson}.

For any positive integers $m$ and $t$ with $t<m$, there exists an $[m,m-t]_q$ MDS code $\mathcal{C}$ for some prime power $q$. Let $\mathcal{F}=\mathcal{C}$ and the column index set $\mathcal{K}={[0,m)\choose t}\times [0,q)^t$, then we have $F=q^{m-t}$ and the corresponding row index matrix $\mathbf{F}$, whose collection of row vectors is $\mathcal{C}$, is an OA$(m,q,m-t)$. If $m-t\geq t$, i.e., $t\leq \frac{m}{2}$, the row index matrix $\mathbf{F}$ is also an OA$_{\lambda}(m,q,t)$ with $\lambda=q^{m-2t}$.  By Remark \ref{remark-caching-OA},  the array $\mathbf{P}$ generated from Construction \ref{construction} is a $(K,F,Z,S)$ PDA with $K={m\choose t}q^{t}$ and $Z/F=1-(\frac{q-1}{q})^{t}$. Furthermore, we obtain the following result.
\begin{theorem}
\label{th-reult-m-2}
Let $m$, $t$ be positive integers with $m\geq 2t$  and $q\geq2$  be a  prime power. If there exists an $[m,m-t]_q$ MDS code $\mathcal{C}$, then there always exists an $({m\choose t}q^t,q^{m-t}, q^{m-t}-(q-1)^tq^{m-2t}, q^m-q^{m-t})$ PDA which can realize a $({m\choose t}q^t,M,N)$ coded caching scheme with $\frac{M}{N}=1-(\frac{q-1}{q})^t$, subpacketization $F=q^{m-t}$ and transmission load $R=q^t-1$. \end{theorem}
\begin{proof}
Let $\mathcal{F}=\mathcal{C}$ and $\mathcal{K}={[0,m)\choose t}\times [0,q)^t$. From the above analysis, the array $\mathbf{P}$ generated from Construction \ref{construction} is a $(K,F,Z,S)$ PDA with $K={m\choose t}q^{t}$, $F=q^{m-t}$ and $Z/F=1-(\frac{q-1}{q})^{t}$. The rest is  to prove $R=q^{t}-1$, i.e., $S=FR=q^m-q^{m-t}$.

First we  prove that each vector in $\mathbb{F}^m_q \setminus \mathcal{C}$ appears in $\mathbf{P}$.   By Construction \ref{construction}, a vector ${\bf e}$ appears in $\mathbf{P}$ if and only if there exists a vector ${\bf f}\in \mathcal{C}$ satisfying $d({\bf e},{\bf f})=t$.   Since $\mathcal{C}$ is an $[m,m-t]_q$ MDS code,  by Lemma \ref{lem-sphere}, we have $\bigcup_{{\bf f}\in \mathcal{C}}\mathcal{E}_{\mathcal{C}}({\bf f},t)=\mathbb{F}^m_q$.  Hence, for any vector ${\bf e}\in \mathbb{F}^m_q \setminus \mathcal{C}$, there exists a vector ${\bf f}\in \mathcal{C}$ satisfying $0<d({\bf e},{\bf f})\leq t$. If $d({\bf e},{\bf f})=t$, then ${\bf e}$ appears in $\mathbf{P}$. If $0<d({\bf e},{\bf f})< t$, there exists a vector ${\bf e}'$ satisfying 1) $d({\bf e}',{\bf f})=t$ and 2) ${\bf e}$ is located on the line generated by ${\bf f}$ and ${\bf e}'$. Consequently we have ${\bf e}=\alpha {\bf f} + \beta {\bf e}'$ for some $\alpha,\beta\in \mathbb{F}_q\setminus\{0\}$. Then we have ${\bf e}=(\alpha+\beta){\bf f} +\beta({\bf e}'-{\bf f})$, which implies $d({\bf e},(\alpha+\beta){\bf f})=wt(\beta({\bf e}'-{\bf f}))=t$. Since $(\alpha+\beta){\bf f}\in \mathcal{C}$, ${\bf e}$ appears in $\mathbf{P}$.

Next, for each vector ${\bf e}$ occurring in $\mathbf{P}$, we show that the number of its occurrences in each column is at most once. Assume that ${\bf e}$ appears at $p_{{\bf f}_1,(\mathcal{T}, {\bf b})}$ and $p_{{\bf f}_2,(\mathcal{T}, {\bf b})}$, then we have ${\bf f}_1|_{[0,m)\setminus \mathcal{T}}={\bf e}|_{[0,m)\setminus \mathcal{T}}={\bf f}_2|_{[0,m)\setminus \mathcal{T}}$ by Construction \ref{construction}. Consequently,  we have ${\bf f}_1={\bf f}_2$, because ${\bf f}_1$ and ${\bf f}_2$ belong to the $[m,m-t]_q$ MDS code $\mathcal{C}$.

Therefore  $S=q^m-q^{m-t}$. The proof is complete.
\end{proof}

From Theorem \ref{th-s=m-1} we have an $({m\choose t}q^t,M,N)$ coded caching scheme, say Scheme1, with memory ratio $\frac{M}{N}=1-(\frac{q-1}{q})^{t}$, subpacketization $F_1=q^{m-1}$ and transmission load $R_1=(q-1)^t$. From Theorem \ref{th-reult-m-2} we have an $({m\choose t}q^t,M,N)$ coded caching scheme, say Scheme2, with memory ratio $\frac{M}{N}=1-(\frac{q-1}{q})^{t}$, subpacketization $F_2=q^{m-t}$ and transmission load $R_2=q^t-1$. Clearly Scheme1 has the same number of users $K$ and memory ratio $\frac{M}{N}$ as Scheme2. Moreover, $\frac{F_1}{F_2}=q^{t-1}$ and $\frac{R_1}{R_2} = \frac{(q-1)^t}{q^t -1}$. If $q$ is large enough,  then $\frac{R_1}{R_2}$ approximates $1$ and $\frac{F_1}{F_2}$ approximates infinity. This implies that  Scheme 2 has much smaller subpacketization and almost the same transmission load compared to Scheme1 for large $q$. For example,
we compare them in Table \ref{tab-compare} when $t=2$ and $m=10$, $20$, $30$ and $40$. From Table \ref{tab-compare}, we can see that for fixed $t$, the memory ratio $\frac{M}{N}$ decreases as $q$ increases. Furthermore,  the value of $\frac{R_1}{R_2}$ approximates $1$  and  the value of $\frac{F_1}{F_2}$  increases fast as $q$ increases.
\begin{table}[http!]
\center
\caption{The comparisons of schemes from Theorem \ref{th-s=m-1} and Theorem \ref{th-reult-m-2}\label{tab-compare}}
\small{
\begin{tabular}{|c|c|c|c|c|c|c|}
\hline
$m$ &$q$ &$K$ &$\frac{M}{N}$   &Rate ratio $\frac{R_1}{R_2}$   &subpacketization ratio $\frac{F_1}{F_2}$    \\ \hline
10& 11&    5445&0.1736         &0.9091                         &11\\ \hline
20& 23& 100510 &0.0930         &0.9167                         &23\\ \hline
30& 31&  418035&0.0635         &0.9375                         &31\\ \hline
40& 41&1311180 &0.0482         &0.9524                         &41\\ \hline
\end{tabular} }
\end{table}
\section{Conclusion}
\label{conclusion}
In this paper, we introduced a  framework to construct PDA and thus coded caching schemes. Based on this framework, we only need to choose the row index matrix and the column index set appropriately in order to design a coded caching scheme. We demonstrated our approach by  providing a generalization of several previously known schemes constructed in \cite{MN,SZG}.
When the number of users is $K={m\choose t}q^t$ for any positive integers $m$, $q$ and $t$ with $t\leq m$, we showed that the row index matrix must be an orthogonal array  if each user has the same memory size.  On the other hand,  when the coded gain is ${m\choose t}$, which is the largest coded gain under this framework,  the row index matrix must be a covering array. Consequently lower bounds on the transmission load $R$ and the subpacketization $F$ were derived.  In particular, we showed that the scheme in \cite{YCTC} achieves these two lower bounds.  Furthermore, We  demonstrated our approach by producing new  explicit schemes with smaller subpacketization than some previously known schemes.
Using a trivial orthogonal array as the row index matrix, we constructed a new scheme, which has the same number of users, memory ratio and transmission load as the scheme in \cite{SZG}, while its subpacketization is smaller.
Similarly,  based on a MDS code, we constructed another new scheme with the same number of users, memory ratio and almost the same transmission load (under certain conditions) as the scheme in \cite{SZG}, while its subpacketization is much smaller. Finally we remark that all the row vectors of $\mathcal{F}$ considered in this paper are all distinct. However, we can also consider the case where some row vectors are repeated using our construction.




%

\ifCLASSOPTIONcaptionsoff
  \newpage
\fi



%
\bibliographystyle{IEEEtran}

\end{document}